\pdfoutput=1
\documentclass[10pt,journal,twocolumn]{IEEEtran}

\usepackage{amsmath,amsfonts}
\usepackage{amsthm}
\usepackage{cite}
\usepackage{graphicx} 
\usepackage{color}
\usepackage{datetime}
\usepackage{epstopdf}
\usepackage{multirow}
\usepackage[font=small,labelfont=bf]{caption}
\usepackage[caption=false,font=footnotesize,labelfont=sf,textfont=sf]{subfig}
\usepackage{url}
\usepackage{array}
\usepackage[normalem]{ulem}
\usepackage{cases}
\usepackage{setspace}

\newtheorem{lemma}{Lemma}
\newtheorem{result}{Result}

\def\round{\textrm{round}}
\def\std{\textrm{std}}
\def\MSE{\textrm{MSE}}
\def\hdr{\textrm{HDR}}

\def\pred{\textrm{pred}}

\def\Q{\textrm{Q}}
\def\iQ{\textrm{iQ}}
\def\sQ{\mathcal{Q}}

\def\H{{\mathbf T}}
\def\v{{\mathbf v}}
\def\u{{\mathbf u}}
\def\p{{\mathbf p}}
\def\I{{\mathbf I}}
\def\J{{\mathbf J}}
\def\r{{\mathbf r}}
\def\m{{\mathbf m}}
\def\Y{{\mathbf Y}}
\def\W{{\mathbf W}}

\def\E{\mathbb{E}}
\def\Z{\mathbb{Z}}
\def\R{\mathbb{R}}

\def\aka{{\textit{aka }}}

\def\eg{{\textit{e.g.}}}
\def\ie{{\textit{i.e.}}}

\definecolor{RawSienna}{cmyk}{0,0.72,1,0.45}

\newcommand*\rfrac[2]{{}^{#1}\!/_{#2{\kern 0.06em}}}

\begin{document}

\title{Impact Analysis of Baseband Quantizer\\on Coding Efficiency for HDR Video}

\author{Chau-Wai~Wong,~\IEEEmembership{Member,~IEEE,}
        Guan-Ming~Su,~\IEEEmembership{Senior Member,~IEEE,}
        and~Min~Wu,~\IEEEmembership{Fellow,~IEEE}
\vspace{-5mm}
\thanks{Copyright (c) 2016 IEEE. Personal use of this material is permitted. However, permission to use this material for any other purposes must be obtained from the IEEE by sending a request to pubspermissions@ieee.org.}%
\thanks{Manuscript received March 8, 2016; revised July 15, 2016; accepted July 16, 2016. Date of publication July xx, 2016; date of current version July xx, 2016. The associate editor coordinating the review of this manuscript and approving it for publication was Prof. Yao Zhao.}%
\thanks{C.-W. Wong and Min Wu are with the Department
of Electrical and Computer Engineering, and the Institute for Advanced Computer Studies, University of Maryland, College Park, MD 20742, USA. This work was initiated when C.-W. Wong was a research intern at Dolby Laboratories in 2014. E-mail: ({\kern 0.06em}cwwong,~minwu{\kern 0.06em})@umd.edu.}
\thanks{G.-M. Su is with Dolby Laboratories, Sunnyvale, CA 94085, USA. E-mail: guanmingsu@ieee.org.}
}

\maketitle

\begin{abstract}
Digitally acquired high dynamic range (HDR) video baseband signal can take 10 to 12 bits per color channel. It is economically important to be able to reuse the legacy 8 or 10-bit video codecs to efficiently compress the HDR video. Linear or nonlinear mapping on the intensity can be applied to the baseband signal to reduce the dynamic range before the signal is sent to the codec, and we refer to this range reduction step as a baseband quantization. We show analytically and verify using test sequences that the use of the baseband quantizer lowers the coding efficiency. Experiments show that as the baseband quantizer is strengthened by 1.6 bits, the drop of PSNR at a high bitrate is up to 1.60\,dB. Our result suggests that in order to achieve high coding efficiency, information reduction of videos in terms of quantization error should be introduced in the video codec instead of on the baseband signal.
\end{abstract}


\pdfoutput=1
\vspace{-2mm}
\section{Introduction}

Realizing more vivid digital videos relies on two main aspects: more pixels and better pixels\cite{Brooks-15, betterPixel-article}. The latter is more important than the former nowadays when the resolution goes beyond the high definition.  
At the signal level, achieving better pixels means adopting a wide color gamut (WCG), and using a high dynamic range (HDR) to represent all colors with small quantization errors \cite{edr_patent-15, implicationHDR-15, towardsHDR-15, std_doc1, std_doc2}.

One efficient color coding standard that keeps the visibility of quantization artifacts to a uniformly small level is the perceptual quantizer (PQ)\cite{SMPTE-ST2084, Miller-PQ}, but it still takes 12 bits to represent all luminance levels.
Economically, it is important to be able to reuse the legacy 8 or
10-bit video codecs such as H.264/AVC\cite{H264_standard} and H.265/HEVC (without range extensions)~\cite{HEVC_standard}
in order to efficiently compress HDR videos.
Linear or nonlinear mapping (\eg, reshaping~\cite{GMS_patent-11, GMS_patent-13a, GMS_patent-13b}) on the intensity can be applied
to the baseband signal to reduce the dynamic range before
the signal is sent to the encoder, and we refer to this range
reduction step as a baseband quantization.
Details of the baseband quantizer can be sent as side information to the decoder
to recover the baseband signal.
Even if a codec supports the dynamic range of a video, 
range reduction can also be motivated by the needs of 
i) saving the running time of the codec via computing numbers in a smaller range, 
ii) handling the event of instantaneous bandwidth shortage as a coding feature provided in VC-1\cite{VC1_H264_book, rao-kim-hwang-14, VC1_standard}, or
iii) removing the color precision that cannot be displayed by old screens. 

Hence, it is important to ask whether reducing the bitdepth for baseband signal is bad for coding efficiency measured in HDR. Practitioners would say ``yes'', but if one starts to tackle this question formally, the answer is not immediately clear as the change of the rate-distortion (RD) performance is non-trivial: reducing the bitdepth for baseband signal while maintaining the compression strength of the codec will lead to a smaller size of encoded bitstream and a larger error measured in HDR.

We approach this problem by establishing the relationship between the strength of the baseband quantizer and the coding efficiency measured in (peak) signal-to-noise ratio [(P)SNR]. 
The (P)SNR measure on video signal stored in the PQ format can be approximately considered perceptually uniform because the PQ is by design a perceptually uniform representation in its signal domain\cite{SMPTE-ST2084, Miller-PQ}.
It is beneficial to first model the problem of quantifying the error in the reconstructed images\cite{transformFoundation-01} as the problem of quantifying the error in the reconstructed residues. We then examine the error of a single quantizer, and arrive at Lemma~2 that serves as a primitive to facilitate the joint analysis on the effects of baseband and codec quantizers with a linear transform.

The paper is organized as follows. In Section~\ref{subsec:primitive}, we simplify the practical HDR video coding pipeline into a theoretically tractable model, and then present the main derivation in Section~\ref{subsec:mainDerivation}. Simulation results are presented in Section~\ref{sec:simu} to validate the derivation, and experimental results on videos are presented in Section~\ref{sec:exp} to confirm the theoretical explanation.

\pdfoutput=1
\section{HDR Video Coding Pipeline Modeling} \label{subsec:primitive}

\begin{figure}[!t]
 \centering
 \vspace{-4mm}
 \includegraphics[width=3.5in]{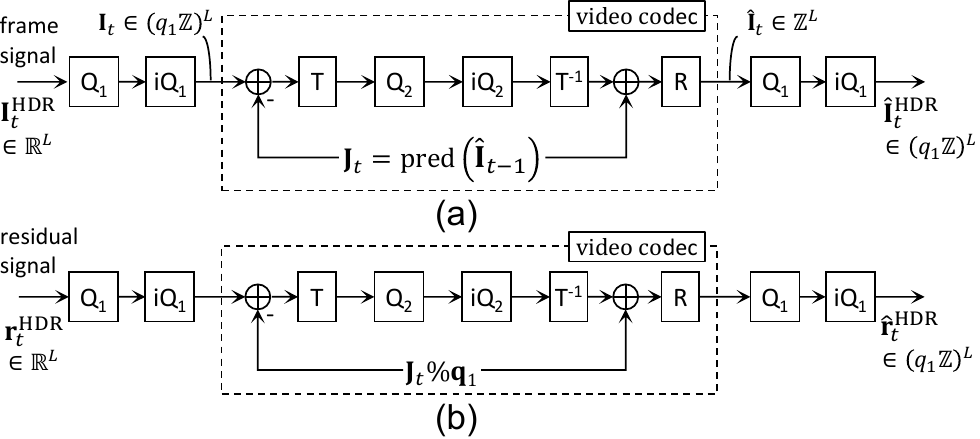}
 \caption{(a) Block diagram for the video coding process
with the effect of baseband signal quantization, and 
(b) equivalent diagram of (a).
Block R is the rounding to the nearest integer operation, $\round(x)$.  
$\Q_i(x) \stackrel{\textrm{def}}{=} \round \left( x / q_i \right),
\textrm{i}\Q_i(x) \stackrel{\textrm{def}}{=}  q_i \cdot x, \ i=1,2$ are quantization and dequantization, respectively.
All operations are applied separately to each entry of $x$ when $x$ is a vector.}
 \label{derive:block_diagram}
\end{figure}		

\subsection{Quantifying Frame Error by Residue Error}
Block diagram shown in Fig.~\ref{derive:block_diagram}~(a) models the video coding pipeline with the effect of baseband signal quantization. The input to the pipeline is the HDR frame at time index~$t$, $\I_t^{\hdr}$, with $L$ pixels. The immediate input to the video codec $\I_t$ and final reconstructed output $\hat{\I}_t^{\hdr}$ are limited by the precision of the finite bits container, so pixels take values on the set ${q_1 \Z} = \{nq_1 | n \in \Z \}$. The immediate output pixels from the codec take integer values due to the rounding operation at the final stage of the codec, and the integer-valued vector $\hat{\I}_{t-1}$ is used by intra- and inter-predictors collectively modeled as $\pred(\cdot)$.

\begin{lemma}[frame error by residue error]
The problem of quantifying the error of predictively coded video frames can be reduced approximately to quantifying the error of non-predictively coded residues.
\end{lemma}

\begin{proof}
For simplicity, define the quantizer function $\sQ_i(x) = \iQ_i \left( \Q_i(x) \right)$.
Denote the predicted frame $\pred(\hat{\I}_{t-1})$ as $\J_t$, and it can be decomposed into the residue vector with the smallest absolute value for each coordinate, and a vector of integer multiples of $q_1$, namely, 
\begin{equation} \label{eq:mc_frame_decomp}
  \J_t = \J_t\%q_1 + \sQ_1(\J_t)	
\end{equation}
where $\%$ is the modulo operation. Following Fig.~\ref{derive:block_diagram}~(a), the error due to the joint effect of baseband quantization and video compression $\hat{\I}_t^\hdr - \I_t^\hdr$ can be written as: 
\begin{equation} \label{eq:pipeline_orig}
\sQ_1 \hspace{-1mm} \left( \H^{-1} \sQ_2 \hspace{-1mm} \left\{ \H \left[ \sQ_1 \hspace{-0.5mm} (\I_t^\hdr) - \J_t \right] \right\} + \J_t \right) - \I_t^\hdr.
\end{equation}
Substituting Eqn.~(\ref{eq:mc_frame_decomp}) into (\ref{eq:pipeline_orig}) and moving $\sQ_1(\J_t) \in ({q_1 \Z})^L$ into and out of the quantizer with step size $q_1$, we obtain:
\begin{multline} \label{eq:equivalent_model}
\sQ_1 \Bigl( \H^{-1} \sQ_2 \hspace{-1mm} \left\{ \H \left[ \sQ_1 \hspace{-0.5mm} (\I_t^\hdr - \sQ_1(\J_t)) - \J_t\%q_1 \right] \right\}  \\
 + \J_t\%q_1 \Bigr) - \left[ \I_t^\hdr - \sQ_1(\J_t) \right].
\end{multline}
Here, $\I_t^\hdr - \sQ_1(\J_t)$ can be considered as an intra- or inter-prediction residue, and we define it as $\r_t^\hdr$. In terms of quantifying error for reconstructed HDR frames, Fig.~\ref{derive:block_diagram}~(a) is therefore equivalent to Fig.~\ref{derive:block_diagram}~(b) visualized from Eqn.~(\ref{eq:equivalent_model}), namely, 
\begin{equation}
\hat{\I}_t^\hdr - \I_t^\hdr = \hat{\r}_t^\hdr - \r_t^\hdr.
\end{equation}
Assuming the quantization step of $\Q_1$ is much smaller when compared to the range of $\r_t^\hdr$, the predictive branch $\J_t\%q_1$ can be removed to obtain a slightly perturbed residue $\tilde{\r}_t^\hdr$. Therefore, the error of non-predictively coded residues $\tilde{\r}_t^\hdr - \r_t^\hdr \approx \hat{\I}_t^\hdr - \I_t^\hdr$. That is, the non-predictive coding branch of Fig.~\ref{derive:block_diagram}~(b) is approximately equivalent to the original pipeline of Fig.~\ref{derive:block_diagram}~(a).
\end{proof}

\begin{figure}[!t]
 \centering
 \vspace{-10mm}
 \subfloat[]{\includegraphics[width=1.5in]{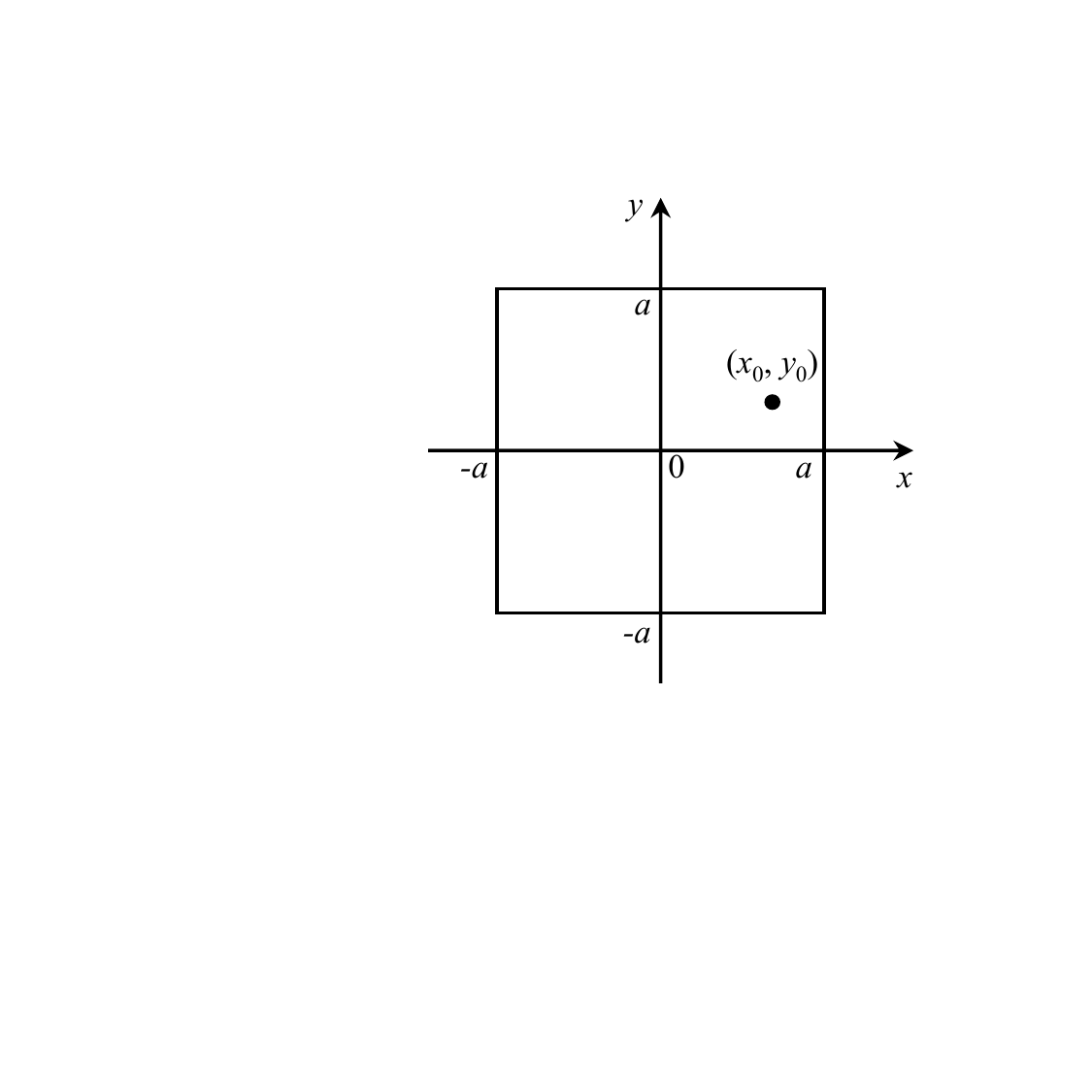}}
 \subfloat[]{\includegraphics[width=1.7in]{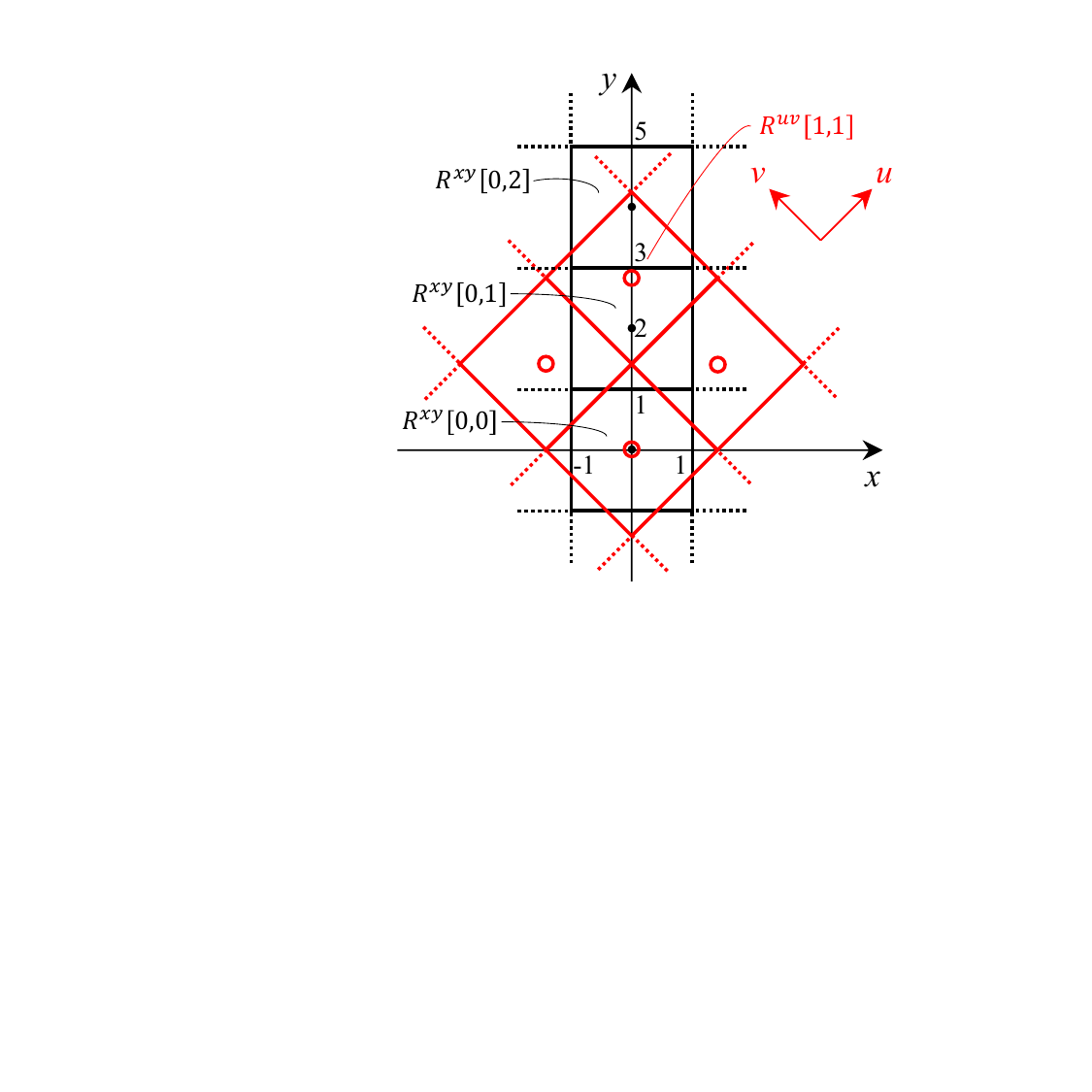}}
 \caption{Illustration for MSE calculation for the cases that (a) any point $(x,y)$ located within the $2a$-by-$2a$ square is quantized to the reconstruction centriod $(x_0,y_0)$ that may or may not located within the square, and (b) a quantization in $xy$-plane is followed by a transform, a quantization in $uv$-plane, inverse transform, and a quantization in $xy$-plane.}
 \label{derive:mse_square}
\end{figure}		

\subsection{Quantization Error for a Hypercube} 
Assume that the reconstruction centroid for a squared region of edge length $2a$ centered at $(0,0)$\footnote{Throughout this paper, column vector $[x_1\ x_2\ \cdots\ x_n]^T$ may be denoted as $(x_1, x_2, \cdots, x_n)$ for the purpose of compact presentation.} as shown in Fig.~\ref{derive:mse_square}~(a) is located at $(x_0,y_0) \in \mathbb{R}^2$, not limited to be within the region. We further assume that the point $(X,Y)$ is uniformly distributed over the square, namely, the joint distribution $f_{X,Y}(x,y) = \frac{1}{4a^2},\ (x,y) \in [-a,a]^2$. The mean-squared error (MSE) for the random vector $(X,Y)$ quantized to/reconstructed at $(x_0,y_0)$ is
\begin{equation}
\begin{split}
  \MSE &= \E \left[ \|(X,Y)-(x_0,y_0)\|^2 \right] \\
	    &= \int_{-a}^{a} \int_{-a}^{a} \|(x,y)-(x_0,y_0)\|^2 \, f_{X,Y}(x,y) \, dx \, dy \\
			&= \frac{1}{4a^2} \int_{-a}^{a} dy \int_{-a}^{a} (x-x_0)^2 + (y-y_0)^2 \, dx \\
			&= d^2 + \frac{2}{3}a^2
\end{split}
\end{equation}
where $d^2 = x_0^2 + y_0^2$ is the squared Euclidean distance to the geometric center of the region, $(0,0)$, and $\frac{2}{3}a^2$ is related to the strength of the quantizer. It is straight forward to extend the result to the $N$-dimensional ($N$-d) case shown as follows:

\begin{lemma}[quantization error]
The mean-squared error for a point that is uniformly distributed within an $N$-d hypercube with an arbitrarily positioned reconstruction centroid and edge length $2a$ is $d^2 + \frac{N}{3}a^2$, where $d$ is the Euclidean distance from the centroid to the geometric center of the hypercube.
\end{lemma}

This result agrees with two intuitive observations. First, as the reconstruction centroid departs from the geometric center, the quantization error increases. Second, as the quantizer strength quantified by the edge length $2a$ increases, the error increases.\\

\section{Effect of Baseband Quantizer\\on Coding Efficiency}

\subsection{Error of Video Coding With Baseband Quantizer} \label{subsec:mainDerivation}
Lemma 1 allows us to avoid dealing with the predictive coding loop in the following analysis, and only to follow a scheme with transform coding and quantization blocks in series. In addition, the residue signal is used as the input as the residue can be more easily modeled in a probabilistic sense than the frame signal. Lemma 2 converts the derivation of the reconstruction error of all possible points to that of just a few reconstruction centriods.

We again use an example with two axes as shown in Fig.~\ref{derive:mse_square}~(b) to illustrate the idea behind, and all the derivations can be easily expanded to the $N$-d general case. 

Assume the input residual signal is a data point $(x,y)$ on the $xy$-plane with a joint probability distribution $f_{XY}(x,y)$. A transform by an orthogonal matrix $\H$ can be considered geometrically as a rotation of the coordinate system, namely,
\begin{align}
(x,y) \stackrel{\H}{\mapsto} (u,v),\ \ \ [u\ v]^T = \H \, [x\ y]^T
\end{align}
\normalsize
where we choose $\H = \frac{1}{\sqrt{2}} \left( \begin{smallmatrix} 1 & 1 \\ -1 & 1 \end{smallmatrix} \right)$. In this example, $(1,0) \stackrel{\H}{\mapsto} (\frac{1}{\sqrt{2}},-\frac{1}{\sqrt{2}})$ and $(1,1) \stackrel{\H}{\mapsto} (\sqrt{2},0)$.

Scalar quantization is equivalent to cutting the plane into squares, as shown in Fig.~\ref{derive:mse_square}~(b). We denote the point set containing all points belonging to a quantized region in $xy$-plane with horizontal index $i$ and vertical index $j$ as $\textsl{R}^{xy}[i,j]$, where indices $i,j \in \mathbb{Z}_M \stackrel{\textrm{def}}{=}\: \{-M, \cdots, 0, \cdots, M\}$. Geometric centers in of $\textsl{R}^{xy}[i,j]$ are denoted as ``$\bullet$'', and those of $\textsl{R}^{uv}[i,j]$ are denoted as ``\textcolor[rgb]{1,0,0}{$\circ$}''. In this example, $\textsl{R}^{xy}[0,1]$ centered at $(0,2)$ and $\textsl{R}^{xy}[0,2]$ centered at $(0,4)$ are both quantized to $\textsl{R}^{uv}[1,1]$ by $\sQ_2$, and finally quantized to $\textsl{R}^{xy}[0,1]$ by $\sQ_1$.

The overall error $D \stackrel{\textrm{def}}{=}\, \E [ \| (x,y)-(\hat{x},\hat{y}) \|^2  ]$ due to video coding with baseband quantizer can be calculated by averaging MSE over the $(2M+1)^2$ regions indexed by $(i,j)$, namely, 
\begin{equation} \label{eq:overall_error}
D = 
\E \Bigl[ \E \bigl[ \| (x,y)-(\hat{x},\hat{y}) \|^2  | \textsl{R}^{xy}[I,J] \bigr] \Bigr]
\end{equation}
where the probability mass function is $p_{IJ}(i,j) = \int_{x,y \in \textsl{R}^{xy}[i,j]} f_{XY}(x,y)\,dxdy$. For each region $\textsl{R}^{xy}[i,j]$, the calculation of error is simplified\footnote{Note that the uniform distribution assumption of Lemma~2 is valid within region $\textsl{R}^{xy}[i,j]$ for the high bitrate coding scenario that we are interested in.} by Lemma~2, namely,
\begin{equation} \label{eq:region_error}
\small
\E \bigl[ \| (x,y)-(\hat{x},\hat{y}) \|^2  | \textsl{R}^{xy}[i,j] \bigr]
= \frac{2}{3} \left( \frac{q_1}{2} \right)^2 + 
 d^2\{\textsl{R}^{xy}[i,j]\}
\end{equation}
where $d\{\textsl{R}^{xy}[i,j]\}$ is the Euclidean distance from the reconstruction centroid to the geometric center of $\textsl{R}^{xy}[i,j]$. The geometric center is by definition $\m = (iq_1, jq_1)$. Passing $\m$ through the whole pipeline shown in Fig.~\ref{derive:block_diagram}~(b) excluding the predictive branch (\aka the main branch), one can obtain the reconstruction centroid:
\begin{subequations}
\begin{align}
\hat{\m} &= \sQ_1 \Big( \H^{-1} \left\{ \sQ_2 \left[ \H \, \sQ_1 \left( [iq_1\ jq_1]^T \right) \right] \right\} \Big) \label{eq:recon_centroid_line1} \\
&= q_1 \; \round \left[
\frac{q_2}{q_1} \, \H^{-1} \, \round \left( \frac{q_1}{q_2} \, \H \begin{bmatrix} i\\ j \end{bmatrix} \right) 
\right]. \label{eq:recon_centroid_line2}
\end{align}
\end{subequations}
Substituting Eqn.~(\ref{eq:region_error}) into 
Eqn.~(\ref{eq:overall_error}), we obtain:
\begin{equation}
D = \frac{2}{3} \left( \frac{q_1}{2} \right)^2 + \E \bigl[ 
 d^2\{\textsl{R}^{xy}[I,J]\}  \bigr].
\end{equation}
Due to the space limitation, we leave the detailed derivation for $\E \left[ d^2\{\textsl{R}^{xy}[I,J]\} \right]$ to the Appendix in the supplementary material. We present the final result of the derivation for the overall error $D$ for scenarios that the baseband quantizer is finer than the codec quantizer (\ie, $q_1 < q_2$) as follows:
\begin{numcases}{D =} \label{eq:recon_center_pos3}
\frac{N}{12} \left( q_2^2 + 2 q_1^2 \right), & \hspace{-6mm} $q_1 \le \frac{q_2}{2}$;\\
\frac{N}{12} \left[ q_2^2 + (1+\gamma_1) q_1^2 + 2 \gamma_{12} q_2 q_1 \right], & \hspace{-6mm} $q_1 > \frac{q_2}{2}$,
\end{numcases}
where $N$ is the length of input signal vectors, and estimates of $\gamma_1$ and $\gamma_{12}$ are displayed in Fig.~\ref{fig:exp_gamma_vs_alpha}~(a).

\begin{figure}[!t]
 \centering
 \vspace{-5mm}
 \subfloat[]{\includegraphics[width=1.8in]{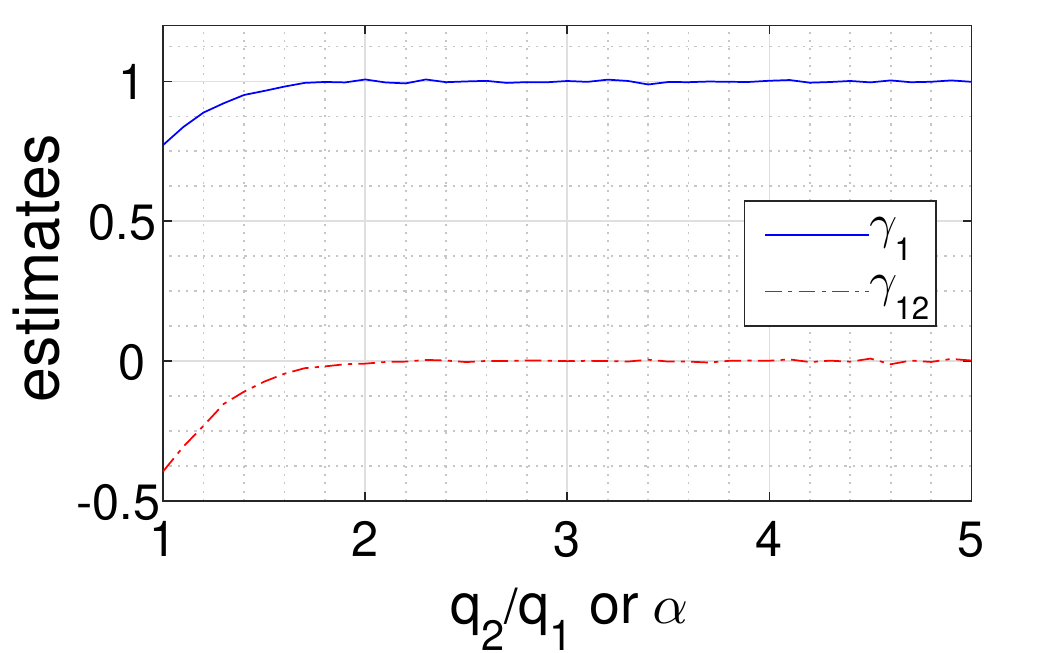}}
 \subfloat[]{\includegraphics[width=1.8in]{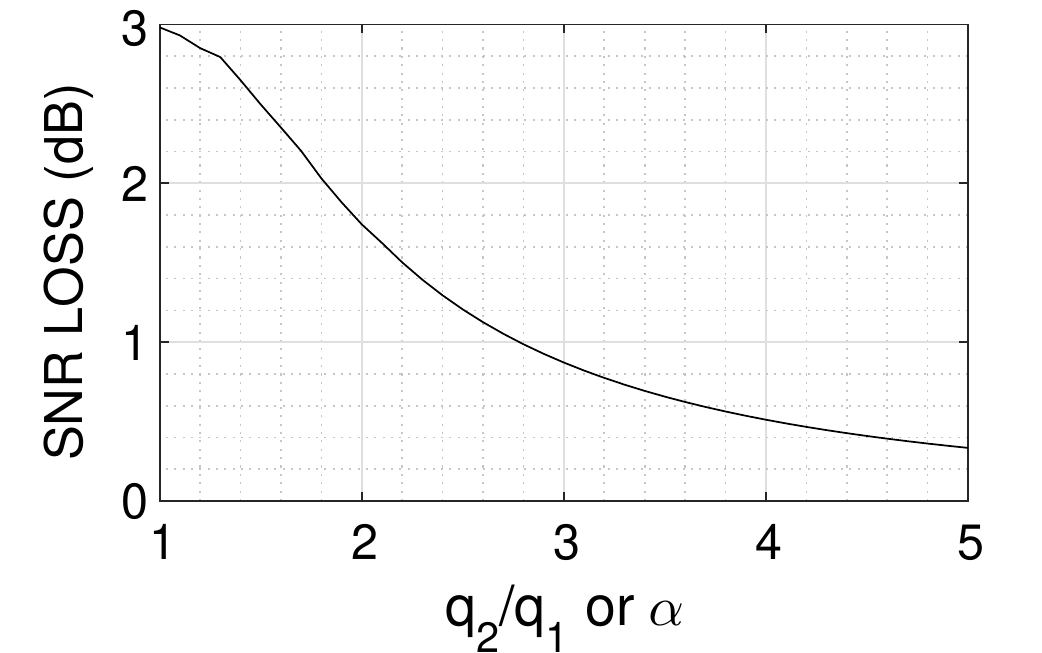}}
 \caption{(a) Estimated $\gamma_1$ and $\gamma_{12}$, and (b) SNR LOSS as a function of $q_2/q_1$ or $\alpha$. Note that for the above curves, we have extended the definitions of $\gamma_1$ and $\gamma_2$ to the case $q_2/q_1 > 2$ in which the true values of $\gamma_1$ and $\gamma_{12}$ are $1$ and $0$, respectively.}
 \label{fig:exp_gamma_vs_alpha}
 \vspace{-2mm}
\end{figure}

\normalsize
It can be proved that, using the scheme of the main branch of Fig.~\ref{derive:block_diagram}~(b), the bitrate is solely controlled by the codec quantizer $\sQ_2$. Hence, fixing $q_2$ and thus the bitrate, any increase in $q_1$ leads to a decrease in SNR and therefore in coding efficiency. In comparison, a change in $q_2$, which changes bitrate and SNR simultaneously, has no impact on the coding efficiency.\footnote{Recall that the comparison of coding efficiency between two codecs is via the comparison of their empirical RD curves. A change in $q_2$ does lead to a move of the operation point in the bitrate-SNR plane, but both the starting and the ending locations reside on the same RD curve.}

Given the scenarios of interest that $q_1 < q_2$, we define $q_1 = \tfrac{q_2}{\alpha}, \ \alpha \ge 1$. The SNR loss with reference to an almost perfectly fine baseband quantizer, \ie, $q_1 \to 0$, can be easily derived:
\small
\begin{numcases}{\hspace{-6mm}\textrm{SNR LOSS} =} \label{eq:snr_loss}
10 \log_{10} \left( 1 + \frac{2}{\alpha^2} \right), &\hspace{-6mm} $\alpha \ge 2$; \\
\label{eq:snr_loss2}
10 \log_{10} \left( 1 + \frac{1+\gamma_1}{\alpha^2} + \frac{2\gamma_{12}}{\alpha} \right), &\hspace{-6mm}  $\alpha < 2$.
\end{numcases}
\normalsize
The resulting SNR loss is shown in Fig.~\ref{fig:exp_gamma_vs_alpha}~(b).

To conclude, under the assumption of $q_1 < q_2$,\ \ i) the best case is $q_1 \ll q_2$ or $\alpha \to \infty$, and error is solely due to the codec quantizer and there is no reduction in SNR; and ii) the worst case is reached when $q_1$ increases to $q_2$ or $\alpha$ decreases to $1$, and a maximum of $3$\,dB SNR drop is incurred.

\begin{figure*}[!t]
 \centering
 \vspace{-13mm}
 \subfloat[]{\includegraphics[width=2.5in]{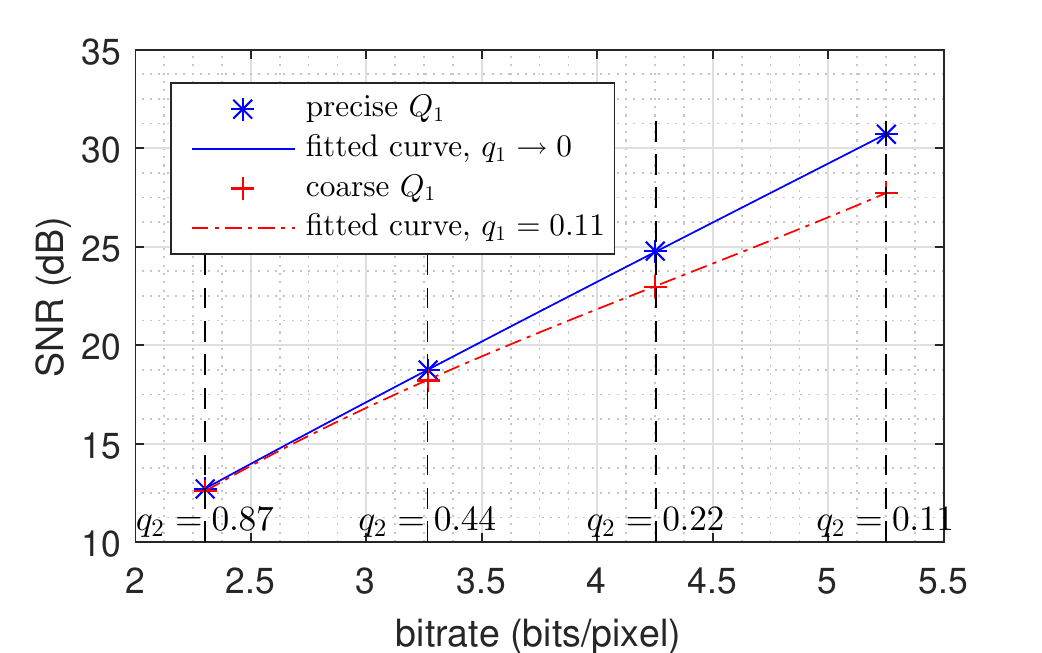}}
 \subfloat[]{\includegraphics[width=2.5in]{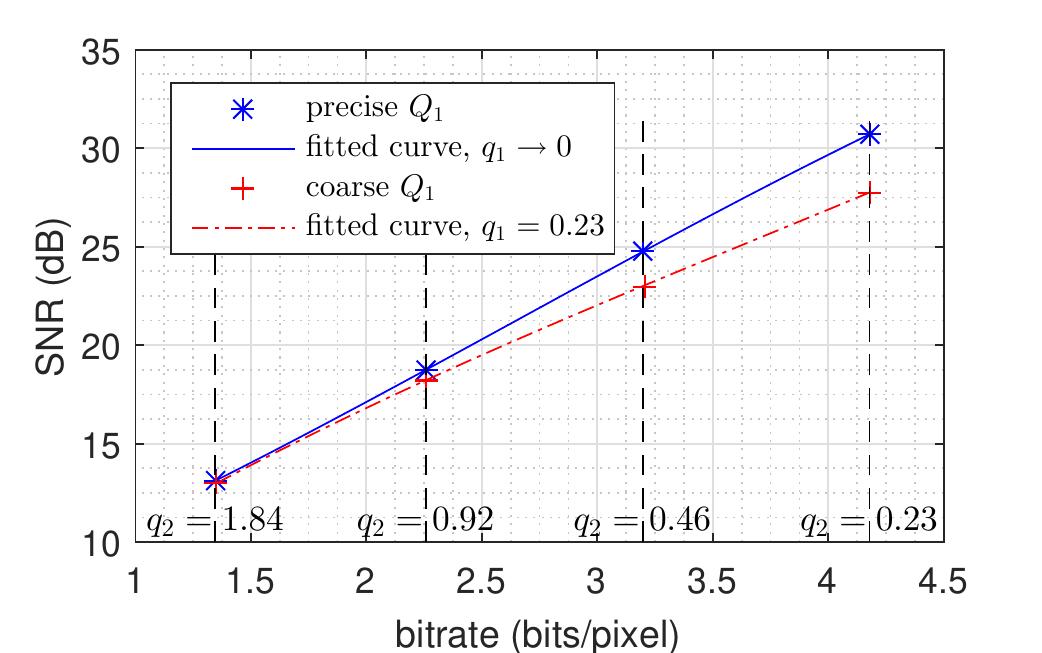}}
 \subfloat[]{\includegraphics[width=2.1in]{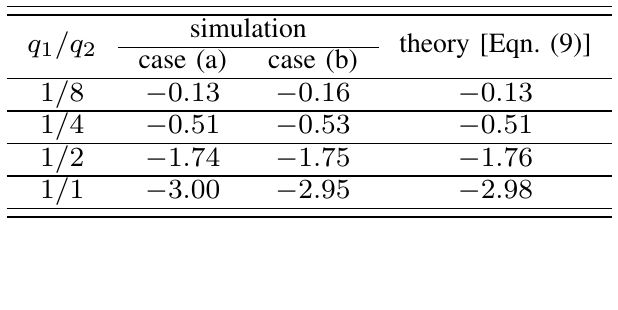}}
 \vspace{-1mm}
 \caption{RD curves of simulated video data revealing the performance gap between the scenarios that baseband quantizer $\sQ_1$ is negligible (\textcolor[rgb]{0,0,1}{\textbf{solid blue}}) and not negligible (\textcolor[rgb]{1,0,0}{\textbf{dash-dot red}}),
case (a): small blocks with low neighborhood correlation ($L = 4^2$, $\rho = 0.4$), and
case (b): large blocks with high correlation ($L = 16^2$, $\rho = 0.9$). (c) Simulated SNR drops for different $\rfrac{q_1}{q_2}$ ratios agree with the theoretical results [Eqns.~(\ref{eq:snr_loss}--\ref{eq:snr_loss2})].}
 \label{fig:simu_results}
\end{figure*}

\begin{figure*}[!t]
 \centering
 \vspace{-7mm}
 \subfloat[]{\includegraphics[width=2.04in]{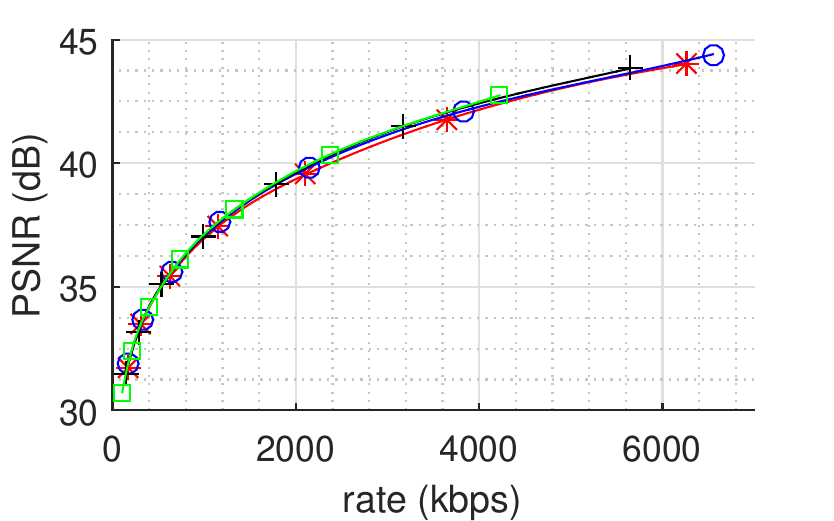}}
 \subfloat[]{\includegraphics[width=2.96in]{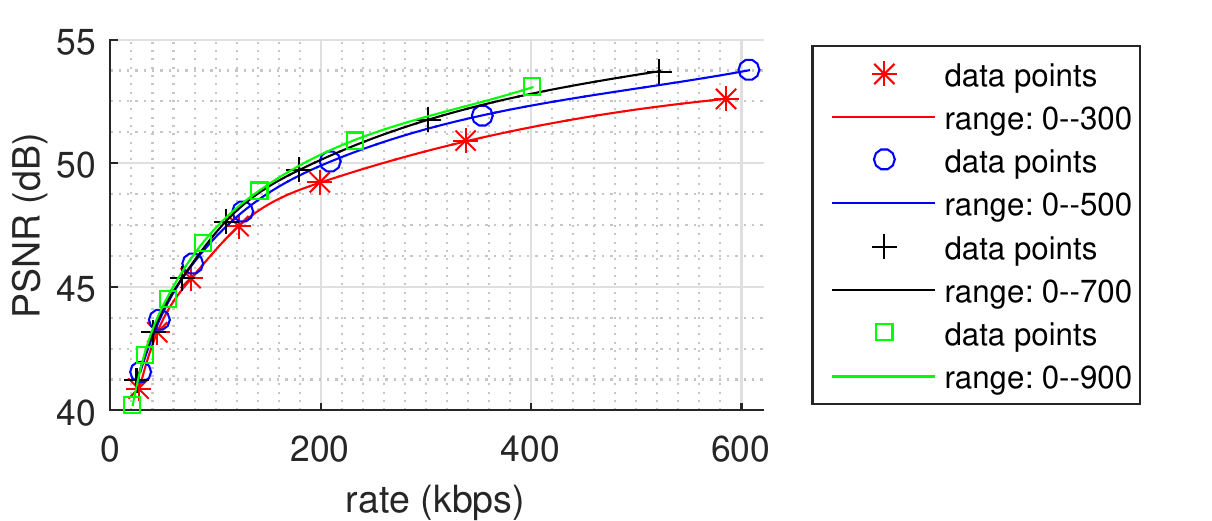}}
 \subfloat[]{\includegraphics[width=2.1in]{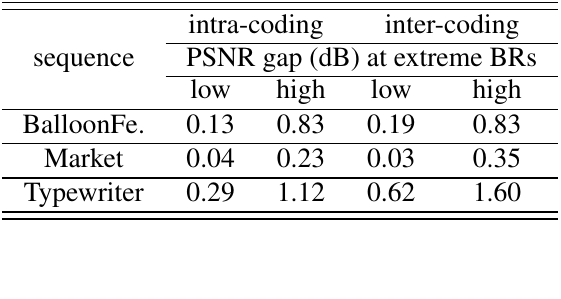}}
 \vspace{-1mm}
 \caption{RD curves for inter-coded videos with different strength of baseband quantizer $\sQ_1$: (a) \textsc{Market} and (b) \textsc{Typewriter}. The PSNR gaps at high bitrate are $0.35$ and $1.60$\,dB, respectively. (c) Largest PSNR gaps at both high and low bitrates for intra- and inter-coded videos.}
 \label{fig:exp_res_fig}
 \vspace{-2mm}
\end{figure*}

\subsection{Simulation Results} \label{sec:simu}

We verify the theoretical result by simulating the change of SNR as a function of $\tfrac{q_1}{q_2}$.
Specifically, assume a length-$L$ Gaussian vector $(X_1, X_2, \cdots, X_L)$, with a fixed correlation of neighboring coordinates, \ie, $\textrm{corr}(X_l, X_{l+1}) = \rho$, for $l = 1, \cdots, L-1$. In image/video coding scenarios, $L$ usually takes value in $\{4^2, 8^2, 16^2\}$.
In our simulation, the realizations of random vectors are obtained by choosing disjoint segments from a realization of an AR(1) process.

We present two simulation cases, namely, (a) for small blocks with low neighborhood correlation ($L = 4^2$, $\rho = 0.4$, $\sigma = \std(X_l) = 1.0911$), and case (b) for large blocks with high correlation ($L = 16^2$, $\rho = 0.9$, $\sigma = 2.2942$). In both cases, we check the performance difference between the scenarios when the baseband quantizer $\sQ_1$ is negligible, \ie, $q_1 \to 0$, and not. When $\sQ_1$ is not negligible, we set the quantizer to be reasonably coarse with respect to the spread of $X_l$, namely, $q_1 = \tfrac{\sigma}{10}$, and the corresponding $q_2$ for each bitrate value on RD curve from left to right are $q_1 \times \{8,4,2,1\}$. 

Simulation results are shown in Fig.~\ref{fig:simu_results}.
The dash-dot red curves show the RD performances when the coarse quantizer $\sQ_1$ is used; the solid blue curves show the performances when $\sQ_1$ is negligible, which is to approximate the scenario that $\sQ_1$ is absent. The RD performance drop with respect to the solid blue curves is
consistent with the theoretical estimates shown in the table of Fig.~\ref{fig:simu_results}~(c). As expected from our theoretical result, the above results are independent of block size $L$ and neighborhood correlation $\rho$.

\pdfoutput=1
\section{Experimental Results on Videos} \label{sec:exp}

We now verify the theoretical results using standard test sequences. Test sequences stored in the 16-bit TIFF container are regarded as the reference/baseband signal. They were first linearly mapped to different dynamic ranges to mimic the effect of the baseband quantizer, the resulting videos were then encoded using HM 14.0~\cite{HM14}, and finally the quality in terms of PSNR was measured in the 16-bit precision.

Detailed simulation conditions are as follows. The luma component of three test sequences \textsc{BalloonFestival}, \textsc{Market}, and \textsc{Typewriter} in BT.2020 color space\cite{rec2020} of size $1920 \times 1080$ are used. The operational bitdepth in the video codec is $10$. Each video is encoded using two structures: the all I-frames structure for $17$ frames (\aka intra-coding), and the IBBB\,$\cdots$ structure for $64$ frames (\aka inter-coding). The codec quantizer takes 6 equally spaced quantization parameters to draw one piece of RD curve. 
Videos are baseband-quantized to the dynamic ranges $[0,300], [0,500], [0,700]$, and $[0,900]$ with effective bitdepth $8.2, 9.0, 9.5$, and $9.8$ bits, respectively. 

The experimental results from all sequences with PSNR measure reveal that the stronger the baseband quantizer is, the more penalty in coding efficiency is incurred. We show the RD performance measured in PSNR for \textsc{Market} and \textsc{Typewriter} that are inter-coded in Figs.~\ref{fig:exp_res_fig}~(a) and (b). The figures reveal that the PSNR gaps become larger as the bitrate increases. More specifically, the PSNR gaps between the green curve and the red curve at a high bitrate (the largest rate that 4 curves simultaneously cover) is $0.35$\,dB for \textsc{Market} and $1.60$\,dB for \textsc{Typewriter}. Table of Fig.~\ref{fig:exp_res_fig}~(c) reports the largest PSNR gaps at both high and low bitrates for intra- and inter-coded videos. 
It is observed that as the baseband quantizer strengthened by $9.8-8.2=1.6$ bits, the drop of PSNR at a high bitrate is up to nearly $1.60$\,dB.

\pdfoutput=1
\section{Conclusion and Discussion}

In this work, we have analyzed the HDR video coding pipeline by explicitly considering the existence of the baseband quantizer. We arrived at the conclusion via both theoretical proof and experiments that the baseband quantizer lowers the coding efficiency, whereas the codec quantizer does not affect the coding efficiency. Hence, information reduction of videos in terms of quantization error should be introduced in the video codec instead of on the baseband signal.

In a more practical scenario, nonlinear mapping is more often used than linear mapping for baseband signal range reduction when the bitdepth is insufficient. Although we have shown that quantizing the baseband signal uniformly leads to a penalty in coding efficiency measured in HDR, it would be interesting to see whether quantizing the baseband signal non-uniformly can also lead to a penalty in coding efficiency.

\newpage
\enlargethispage{-5.7in}
\bibliographystyle{IEEEbib_noURL}
\bibliography{spl_main}

\begin{figure*}
\noindent{}\textbf{Note: This is a supplementary file for paper ``Impact Analysis of Baseband Quantizer on Coding Efficiency for HDR Video'', published in \textit{IEEE Signal Processing Letters} in 2016 by Chau-Wai~Wong, Guan-Ming~Su, and Min~Wu.}
\vspace{5mm}
\end{figure*}

\pagenumbering{gobble}

\newpage

\enlargethispage{-4.5in}
\appendices

\section{Derivation for $\E \left[ d^2\{\textsl{R}^{xy}[I,J]\} \right]$}

\pdfoutput=1
\noindent{}Define a residue function $g(x) = \round(x) - x$, where $g(x) \in (-\frac{1}{2},\frac{1}{2}]$ for any $x>0$, and $[-\frac{1}{2},\frac{1}{2})$ for any $x<0$. Hence, Eqn.~(\ref{eq:recon_centroid_line2}) can be simplified to the sum of three terms:
\begin{multline} \label{eq:recon_center_pos}
\hat{\m} = 
q_1 \begin{bmatrix} i\\ j \end{bmatrix} +
q_2 \ \H^{-1} g \left( \frac{q_1}{q_2} \ \H \begin{bmatrix} i\\ j \end{bmatrix} \right) \\ +
q_1 \ g \left\{
\frac{q_2}{q_1} \ \H^{-1} g \left( \frac{q_1}{q_2} \ \H \begin{bmatrix} i\\ j \end{bmatrix} \right)
\right\}.
\end{multline}
Denote the $n$th row and column of matrix $\H$ by $\v_n^T$ and $\u_n$, respectively. Define $\p = (i,j)$, $Y_n = g \left( \frac{q_1}{q_2} \, \v_n^T \p \right)$, and $W_n = g \left( \frac{q_2}{q_1} \, \u_n^T \Y \right)$. The squared distance is
\begin{subequations} \label{eq:recon_center_pos2}
\begin{align}
d^2\{\textsl{R}^{xy}[I,J]\} &= \left\| \m - \hat{\m} \right\|^2 \\
&= \left\| q_2 \, \H^{-1} \begin{bmatrix} Y_1\\ Y_2 \end{bmatrix} + q_1 \begin{bmatrix} W_1\\ W_2 \end{bmatrix} \right\|^2 \\
&= q_2^2 \left\| \Y \right\|^2 + q_1^2 \left\| \W \right\|^2 + 2q_2q_1 \Y^T \H \W.
\end{align}
\end{subequations}
Since vector $\p \in \mathbb{Z}_M^2$, and the term $\frac{q_1}{q_2} \v_n^T \p$ can take values on a non-degenerated subset of $\mathbb{R}$, except in very rare cases with a certain combination of $q_1, q_2, \v_n$ the term takes value on a subset of $\mathbb{Z}$. For the non-degenerated case, it can be proved that $Y_n$ is approximately uniformly distributed on $(-\frac{1}{2},\frac{1}{2})$. Therefore, $\E [ \left\| \Y \right\|^2 ] = \frac{2}{12}$.

When $q_1 < \frac{q_2}{2}$, the range of every coordinate of $\frac{q_2}{q_1}\Y$ is larger than $(-1,1)$. It can be proved that, $W_n$ is uniformly distributed on $(-\frac{1}{2},\frac{1}{2})$, and $\W$ and $\Y$ are uncorrelated. Therefore, $\E [ \left\| \W \right\|^2 ] = \frac{2}{12}$, and $\E [ \Y^T \H \W ] = \textrm{trace} \{ \H \, \E [ \W \Y^T ] \}$ = 0.

When $q_1 > \frac{q_2}{2}$, as $q_1$ increases, $W_n$ becomes more depend on $\Y$. Statistics $\gamma_1 = \frac{12}{N} \, \E [ \left\| \W \right\|^2 ]$ and $\gamma_{12} = \frac{12}{N} \, \E [ \Y^T \H \W ] = \frac{12}{N} \, \textrm{trace} \{ \H \, \E [ \W \Y^T ] \}$ are empirically evaluated using the Monte Carlo method, and resulting estimates are shown in terms of curves in Fig.~\ref{fig:exp_gamma_vs_alpha}~(a).

Therefore,
\begin{multline} \label{eq:region_error_formula_2d}
\E \bigl[ d^2\{\textsl{R}^{xy}[I,J]\} \bigr] = \\
\begin{cases}
\left( q_2^2 + q_1^2 \right) /\, 6, & 0 < q_1 \le \frac{q_2}{2}; \\
\left( q_2^2 + \gamma_1 q_1^2 + 2 \gamma_{12} q_2 q_1 \right) /\, 6, & \frac{q_2}{2} < q_1 \le q_2.
\end{cases}
\end{multline}
And it is not difficult to generalize the above result to the $N$-d scenario as follows:
\begin{multline} \label{eq:region_error_formula_Nd}
\E \bigl[ d^2\{\textsl{R}^{xy}[I_1, \cdots, I_N]\} \bigr] = \\
\begin{cases}
N \left( q_2^2 + q_1^2 \right) /\, 12, & 0 < q_1 \le \frac{q_2}{2}; \\
N \left( q_2^2 + \gamma_1 q_1^2 + 2 \gamma_{12} q_2 q_1 \right) /\, 12, & \frac{q_2}{2} < q_1 \le q_2.
\end{cases}
\end{multline}

\begin{figure*}
\noindent{}\textbf{Note: This is a supplementary file for paper ``Impact Analysis of Baseband Quantizer on Coding Efficiency for HDR Video'', published in \textit{IEEE Signal Processing Letters} in 2016 by Chau-Wai~Wong, Guan-Ming~Su, and Min~Wu.}
\vspace{5mm}
\end{figure*}

\newpage

\enlargethispage{-1.4in}

\section{Summary of the Theoretical Results}
\pdfoutput=1
The description of a theoretical problem in this appendix is abstracted from the practical video coding problem motivated in the main body of the paper where the quantization effect in raw signal domain cannot be ignored. This is particularly true nowadays for the high-dynamic-range (HDR) video signal that takes 10 to 12 bits comparing to 8 bits in the conventional case.

\vspace{4mm}We refer to the quantizer for the raw video signal as the baseband quantizer, $\sQ_1$, and the quantizer in the video codec as the codec quantizer, $\sQ_2$. We are interested in the cases that the strength of $\sQ_1$ is less than that of $\sQ_2$. Block diagram for the setup that we have investigated in the main body of the paper is shown in Fig.~\ref{fig:theo_blkDiagrams}~(b). Fig.~\ref{fig:theo_blkDiagrams}~(a) serves as a stepping stone for obtaining the result for Fig.~\ref{fig:theo_blkDiagrams}~(b).

\vspace{4mm}In Fig.~\ref{fig:theo_blkDiagrams}, random vector $\u \in \R^N$ is the input signal and $\hat{\u} \in \R^N$ is a reconstructed signal. $\sQ_1$ and $\sQ_2$ are high-rate scalar quantizers with quantization steps $q_1<q_2$. $\H$ is a (non-degenerated) orthogonal transform of size $N$-by-$N$. Entropy coding is by configuration carried out after block $\sQ_2$, so the bitrate is solely controlled by $\sQ_2$. In the comparison below, $\sQ_2$ is fixed, hence the MSE and SNR values are all fairly compared at the same bitrate.

\vspace{4mm}\begin{result}[One baseband quantizer and one codec quantizer]
\begin{equation}
\E \left[ \|u-\hat{u}\|^2 \right] = \frac{N}{12} \left( q_2^2 + q_1^2 \right),\hspace{0.5cm} q_1 \le q_2.
\end{equation}
Given the assumption $q_1 < q_2$, we define $q_1 = \tfrac{q_2}{\alpha}, \ \alpha \ge 1$. The SNR loss with reference to an almost perfectly fine baseband quantizer, \ie, $q_1 \to 0$, is

\begin{equation}
\textrm{SNR LOSS} = 
10 \log_{10} \left( 1 + \frac{1}{\alpha^2} \right),\hspace{0.5cm} \alpha \ge 1.
\end{equation}
\end{result}

\vspace{4mm}\begin{result}[Two baseband quantizers and one codec quantizer]
\begin{equation}
\small
\E \left[ \|u-\hat{u}\|^2 \right] = \begin{cases}
\frac{N}{12} \left( q_2^2 + 2 q_1^2 \right), & 0 < q_1 \le \frac{q_2}{2}; \\
\frac{N}{12} \left[ q_2^2 + (1+\gamma_1) q_1^2 + 2 \gamma_{12} q_2 q_1 \right], & \frac{q_2}{2} < q_1 \le q_2.
\end{cases}
\end{equation}

\begin{equation}
\small
\textrm{SNR LOSS} = 
\begin{cases}
10 \log_{10} \left( 1 + \frac{2}{\alpha^2} \right), & \alpha \ge 2; \\
10 \log_{10} \left( 1 + \frac{1+\gamma_1}{\alpha^2} + \frac{2\gamma_{12}}{\alpha} \right), & 1 \le \alpha < 2 .
\end{cases}
\end{equation}
where estimates of $\gamma_1$ and $\gamma_{12}$ are displayed in Fig.~\ref{fig:theoretical_result}~(a). The resulting SNR loss is shown in Fig.~\ref{fig:theoretical_result}~(b). Simulation using AR(1) signal agrees with the theoretical result.
\end{result}

\begin{figure}[t]
 \centering
 \includegraphics[width=3in]{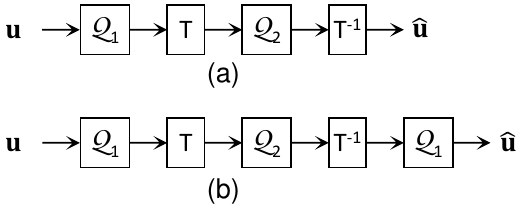}
 \caption{Block diagrams for (a) scenario 1: one baseband quantizer and one codec quantizer, and (b) scenario 2: two baseband quantizers and one codec quantizer.}
 \label{fig:theo_blkDiagrams}
\end{figure}

\begin{figure}[!t]
 \centering
 \subfloat[]{\includegraphics[width=2.2in]{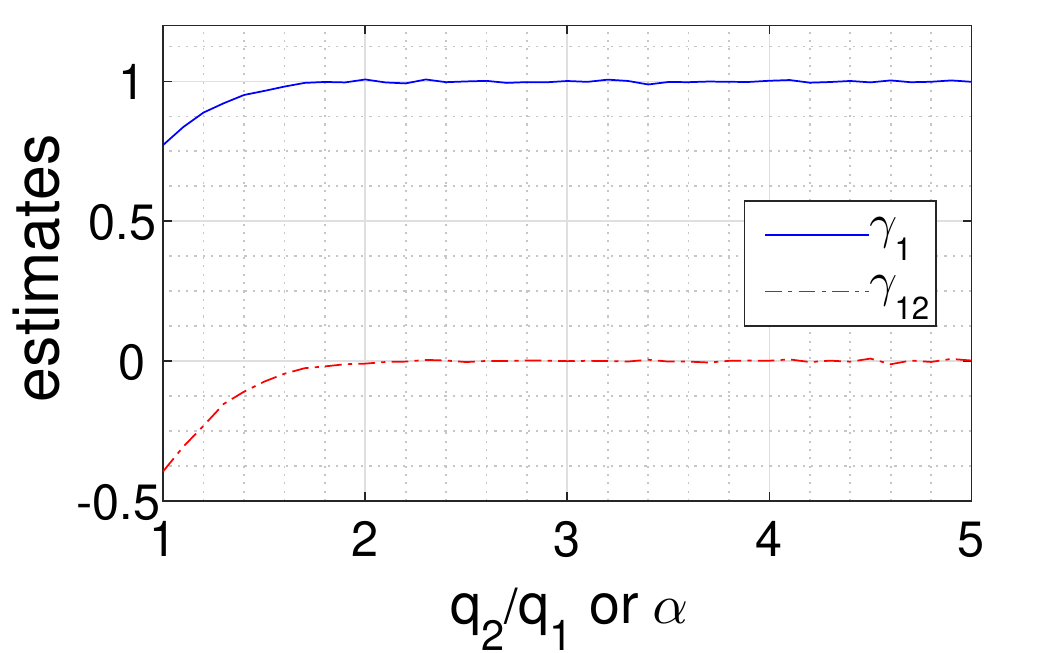}}\\
 \subfloat[]{\includegraphics[width=2.2in]{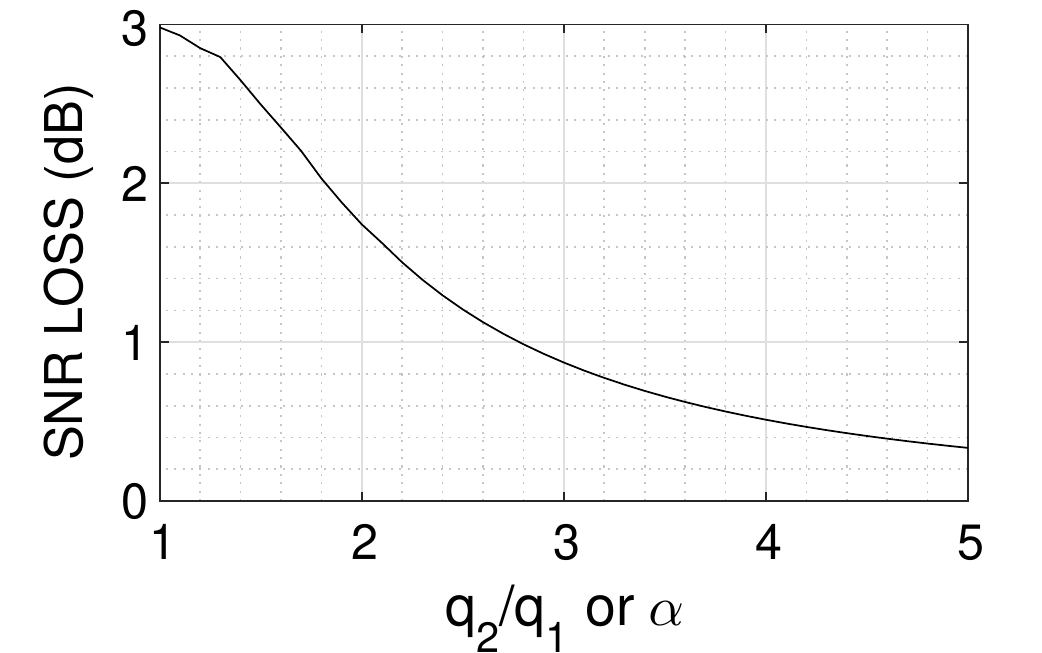}}
 \caption{(a) Estimated $\gamma_1$ and $\gamma_{12}$, and (b) SNR LOSS as a function of $q_2/q_1$ or $\alpha$. Note that for the above curves, we have extended the definitions of $\gamma_1$ and $\gamma_2$ to the case $q_2/q_1 > 2$ in which the true values of $\gamma_1$ and $\gamma_{12}$ are $1$ and $0$, respectively.}
 \label{fig:theoretical_result}
\end{figure}

\end{document}